\documentclass[11pt,a4paper]{article}

\usepackage{fullpage}
\usepackage{color}
\usepackage{amssymb}
\usepackage{amsmath}
\usepackage{amsthm}
\usepackage{graphicx}
\usepackage{algorithm}
\usepackage{algorithmicx}
\usepackage{algpseudocode}
\usepackage{enumerate}

\graphicspath{{./}{figure/}}

\newtheorem{theorem}{Theorem}
\newtheorem{lemma}[theorem]{Lemma}

\newtheorem{corollary}[theorem]{Corollary}

\newcommand{\cancel}[1] {}
\newcommand{\eps}{\varepsilon}
\newcommand{\real}{\mathbb{R}}

\DeclareMathOperator{\conv}{\mathrm{conv}}
\DeclareMathOperator{\seg}{\mathrm{SEG}}

\title{Simplification of Trajectory Streams\thanks{This research is supported by the Research Grants Council, Hong Kong, China (project no.~16208923).}}

\author{Siu-Wing Cheng\thanks{Department of Computer Science and Engineering, HKUST, Clear Water Bay, Hong Kong. Email: scheng@cse.ust.hk. ORCID: https://orcid.org/0000-0002-3557-9935}
\and 
Haoqiang Huang\thanks{Department of Computer Science and Engineering, HKUST, Clear Water Bay, Hong Kong. Email: haoqiang.huang@connect.ust.hk.  ORCID: https://orcid.org/0000-0003-1497-6226}
\and
Le Jiang\thanks{USTC, Hefei, China.  Email: lejiang@mail.ustc.edu.cn.  ORCID: https://orcid.org/0009-0009-8379-0266}}

\begin{document}
	
\maketitle
	
\begin{abstract}
While there are software systems that simplify trajectory streams on the fly, few curve simplification algorithms with quality guarantees fit the streaming requirements.  We present streaming algorithms for two such problems under the Fr\'{e}chet distance $d_F$ in $\mathbb{R}^d$ for some constant $d \geq 2$.  

Consider a polygonal curve $\tau$ in $\mathbb{R}^d$ in a stream.  We present a streaming algorithm that, for any $\varepsilon\in (0,1)$ and $\delta > 0$, produces a curve $\sigma$ such that $d_F(\sigma,\tau[v_1,v_i])\le (1+\varepsilon)\delta$ and $|\sigma|\le 2\,\mathrm{opt}-2$, where $\tau[v_1,v_i]$ is the prefix in the stream so far, and
$\mathrm{opt} = \min\{|\sigma'|: d_F(\sigma',\tau[v_1,v_i])\le \delta\}$.  Let $\alpha = 2(d-1){\lfloor d/2 \rfloor}^2 + d$.  The working storage is $O(\varepsilon^{-\alpha})$.  Each vertex is processed in $O(\varepsilon^{-\alpha}\log\frac{1}{\eps})$ time for $d \in \{2,3\}$ and $O(\varepsilon^{-\alpha})$ time for $d \geq 4$ .  Thus, the whole $\tau$ can be simplified in $O(\eps^{-\alpha}|\tau|\log\frac{1}{\eps})$ time.  Ignoring polynomial factors in $1/\eps$, this running time is a factor $|\tau|$ faster than the best static algorithm that offers the same guarantees.

We present another streaming algorithm that, for any integer $k \geq 2$ and any $\varepsilon \in (0,\frac{1}{17})$, maintains a curve $\sigma$ such that $|\sigma| \leq 2k-2$ and $d_F(\sigma,\tau[v_1,v_i])\le (1+\varepsilon) \cdot \min\{d_F(\sigma',\tau[v_1,v_i]): |\sigma'| \leq k\}$, where $\tau[v_1,v_i]$ is the prefix in the stream so far.  The working storage is $O((k\varepsilon^{-1}+\varepsilon^{-(\alpha+1)})\log \frac{1}{\varepsilon})$.  Each vertex is processed in $O(k\varepsilon^{-(\alpha+1)}\log^2\frac{1}{\varepsilon})$ time for $d \in \{2,3\}$ and $O(k\varepsilon^{-(\alpha+1)}\log\frac{1}{\eps})$ time for $d \geq 4$.
\end{abstract}

\section{Introduction} 

The pervasive use of GPS sensors has enabled tracking of moving objects.  For example, a car fleet can use real-time information about its vehicles to deploy them in response to dynamic demands.   The sensor periodically samples the 
location of the moving object and sends it as a vertex to the remote cloud server for storage and processing.  The remote server interprets the sequence of vertices received as a polygonal curve (trajectory).

For a massive stream, it has been reported~\cite{CWT2006,CXF2012,LMZWH2017,LZSSKJ2015} that sending all vertices uses too much network bandwidth and storage at the server, and it may aggravate issues like out-of-order and duplicate data points.  Software systems have been built for the sensor side to simplify trajectory streams on the fly (e.g.~\cite{LJMZH2019,LMJHW2021,LMZWH2017,WLCZ2021,ZDY2018}).  A local buffer is used.  Every incoming vertex in the stream triggers a new round of computation that uses only the incoming vertex and the data in the local buffer.  At the end of a round, these systems may determine the next vertex and send it to the cloud server, but they may also send nothing.  

It is important to keep the local buffer small~\cite{LJMZH2019,LMJHW2021,LMZWH2017,WLCZ2021,ZDY2018}.  It means a small working storage and less computation in processing the next vertex in the stream, which enables real-time response and requires less computing power.  In theoretical computer science, streaming algorithms have three goals~\cite{M2005}.  Let $n$ be the number of items in the stream so far.  First, the working storage should be $o(n)$ if not $\mathrm{polylog}(n)$.  Second, an item in the stream should be processed in $\mathrm{polylog}(n)$ time or at worst $o(n)$ time because items are arriving continuously.  Third, although the best solution is unattainable as the entire input is not processed as a whole, the solution should be provably satisfactory.  Let $\tau$ be the curve in the stream.  The abovementioned software systems only deal with some local error measures between the simplified curve $\sigma$ and $\tau$. 
 We want to provide guarantees on the size of $\sigma$ and the global similarity between $\sigma$ and $\tau$.

A popular similarity measure is the \emph{Fr\'echet distance}~\cite{Godau1991ANM}. Let $\rho_\tau: [0,1]\rightarrow \real^d$ be a \emph{parameterization} of $\tau$ such that, as $t$ increases from 0 to 1, $\rho_\tau(t)$ moves from the beginning of $\tau$ to its end without backtracking.  It is possible that $\rho_\tau(t) = \rho_\tau(t')$ for two distinct $t, t' \in [0,1]$.   We can similarly define a parameterization $\rho_\sigma$ for $\sigma$. These parameterizations induce a \emph{matching} $\mathcal{M}$ between $\rho_\sigma(t)$ and $\rho_\tau(t)$ for all $t \in [0,1]$.  A point may have multiple matching partners.  The Fr\'echet distance is $d_F(\sigma,\tau) = \inf_{\mathcal{M}} \max_{t \in [0,1]} d(\rho_\sigma(t),\rho_\tau(t))$. We call a matching that realizes the Fr\'echet distance a \emph{Fr\'echet matching}.
	
We study two problems.  The \emph{$\delta$-simplification problem} is to compute, for a given $\delta > 0$, a curve $\sigma$ of the minimum size (number of vertices) such that $d_F(\sigma,\tau) \leq \delta$.   The \emph{$k$-simplification problem} is to compute, for a given integer $k \geq 2$, a curve $\sigma$ of size $k$ that minimizes $d_F(\sigma,\tau)$. 

Given a polygonal curve $\tau = (v_1,v_2,\ldots)$, $|\tau|$ denotes its number of vertices, and for any $i < j$, $\tau[v_i, v_j]$ denotes the subcurve $(v_i,v_{i+1},\ldots,v_j)$.  For a subset $P \subset \real^d$, $d(x, P) = \inf_{p\in P} d(x, p)$.  Given two points $x, y \in \real^d$, $xy$ denotes the oriented line segment from $x$ to $y$.

\subsection{Related works} 

\noindent {\bf Static $\pmb{\delta}$-simplification.} Let $n = |\tau|$.  In $\real$, if the vertices of $\sigma$ are restricted to lie on $\tau$ (anywhere), then $|\sigma|$ can be minimized in $O(n)$ time~\cite{van2019global}; if the vertices of $\sigma$ must be a subset of those of $\tau$, it takes $O(n)$ time to compute $\sigma$ such that $|\sigma|$ is at most two more than the minimum~\cite{driemel2016clustering}.  In $\real$, there is a data structure~\cite{VO2023} that, for any given $\delta$, reports a curve $\sigma$ with the minimum $|\sigma|$ such that the vertices of $\sigma$ lie on $\tau$ (anywhere) and $d_F(\sigma,\tau) \leq \delta$.  The query time is $O(|\sigma|)$, and the preprocessing time is $O(n)$.

In $\real^2$, a curve $\sigma$ with minimum $|\sigma|$ such that $d_F(\sigma,\tau) \leq \delta$ can be computed in $O(n^2\log^2 n)$  time~\cite{guibas1993approximating}. In $\real^d$ for $d \geq 3$, if the vertices of $\sigma$ must be a subset of those of $\tau$, then $|\sigma|$ can be minimized in $O(n^3)$ time~\cite{bringmann2019polyline,van2019global,van2018optimal}. If there is no restriction on the vertices of $\sigma$, it has been recently shown that $|\sigma|$ can be minimized in $O\bigl((|\sigma|n)^{O(d|\sigma|)}\bigr)$ time~\cite{cheng2023solving}.
	
Faster algorithms have been developed by allowing inexact output.  Let $\kappa(\tau, r)$ be the minimum simplified curve size for an error $r$.  In $\real^2$, it was shown~\cite{agarwal2005near} that for any $\delta > 0$, a curve $\sigma$ can computed in $O(n\log n)$ time such that $d_F(\sigma,\tau)\le \delta$ and $|\sigma|\le \kappa(\tau, \delta/2)$.  It is possible that $\kappa(\tau, \delta/2) \gg \kappa(\tau, \delta)$.  A better control of the curve size has been obtained later.  In $\real^d$, it was shown~\cite{van2019global} that for any $\eps \in (0,1)$ and $\delta > 0$, a curve $\sigma$ can be constructed in $O(\eps^{2-2d}n^2\log n\log\log n\bigr)$ time such that $d_F(\sigma, \tau)\le (1+\varepsilon)\delta$ and $|\sigma|\le 2\kappa(\tau, \delta)-2$.\footnote{In~\cite{van2019global}, the first and last vertices of $\sigma$ are restricted to be the same as those of $\tau$.}  Later, it was shown~\cite{cheng2023curve} that for any $\alpha,\varepsilon\in (0,1)$ and $\delta > 0$, a curve $\sigma$ can be computed in $\tilde{O}\bigl(n^{O(1/\alpha)}\cdot(d/(\alpha\varepsilon))^{O(d/\alpha)}\bigr)$ time such that $d_F(\sigma, \tau)\le (1+\varepsilon)\delta$ and $|\sigma|\le (1+\alpha)\cdot\kappa(\tau, \delta)$.  

The algorithms above for $d \geq 2$ do not fit the streaming setting~\cite{agarwal2005near,bringmann2019polyline,cheng2023curve,cheng2023solving,guibas1993approximating,van2019global,van2018optimal}. In~\cite{agarwal2005near}, the simplified curve $\sigma$ is a subsequence of the vertices in $\tau$, and the vertices of $\sigma$ are picked from $\tau$ in a traversal of $\tau$.  Given the last pick $v_i$, for any $j > i$, whether $v_j$ should be included in $\sigma$ is determined by examining the subcurve $\tau[v_i,v_j]$.  This requires an $O(n)$ working storage which is too large.  If the algorithms in~\cite{bringmann2019polyline,cheng2023curve,cheng2023solving,guibas1993approximating,van2019global,van2018optimal} are used in the streaming setting, the processing time per vertex would be
$O(n)$ or more which is too high.

\vspace{5pt}

\noindent {\bf Static $\pmb{k}$-simplification.}  In $\real$, for any given $k$, the data structure for $\delta$-simplification in~\cite{VO2023} can be queried in $O(k)$ time to report a curve $\sigma$ such that $|\sigma|=k$, the vertices of $\sigma$ lie on $\tau$ (anywhere), and $d_F(\sigma,\tau)$ is minimized.  In $\real^d$, if the vertices of $\sigma$ must be a subset of those of $\tau$, it was shown~\cite{Godau1991ANM} that a solution $\sigma$ can be computed in $O(n^4\log n)$ time such that $d_F(\sigma,\tau) \le 7 \cdot \min\bigl\{d_F(\sigma',\tau) : |\sigma'| \leq k,\,\text{no restriction on the vertices of $\sigma'$}\bigr\}$.  The factor 7 can be reduced to 4 by a better analysis~\cite{agarwal2005near}.  Nevertheless, if this algorithm is used
in the streaming setting, the vertex processing time would be $O(n^3\log n)$ or more, which is way too high.

\vspace{5pt}

\noindent {\bf Streaming line simplification.}  A \emph{line simplification} $\sigma$ of a stream $\tau$ is a subsequence of the vertices such that the first and last vertices of $\sigma$ are the first and last vertices in the stream so far.  In $\real^2$, for any integer $k \geq 2$ and any $\eps \in (0,1)$, it was shown~\cite{abam2007streaming,abamDCG} how to maintain $\sigma$ such that $|\sigma| \leq 2k$ and $d_F(\sigma,\tau) \leq (4\sqrt{2}+\varepsilon)\cdot\mathrm{opt}_k$, where $\mathrm{opt}_k = \min \{d_F(\sigma',\tau): \text{line simplification $\sigma'$ of $\tau$}, \, |\sigma'| \leq k\}$. The working storage is $O(k^2 + k\eps^{-1/2})$.  Each vertex is processed in $O(k\log\frac{1}{\eps})$ amortized time.  There are also streaming software systems for this problem that minimize some local error measures (e.g.~\cite{GFX2024,MHPLPR2011,MOHLR2014,PPS2006}); however, they do not provide any guarantee on the global similarity between $\tau$ and the output curve.

\vspace{5pt}

\noindent {\bf Discrete Fr\'{e}chet distance.}  The \emph{discrete} Fr\'{e}chet distance $d_{\text{dF}}(\sigma,\tau)$ is obtained with the restriction that the parameterizations $\rho_\sigma$ and $\rho_\tau$ must match the vertices of $\sigma$ to those of $\tau$ and vice versa. Note that $d_{\text{dF}}(\sigma,\tau) \geq d_F(\sigma,\tau)$ and $d_{\text{dF}}(\sigma,\tau) \gg d_F(\sigma,\tau)$ in the worst case. 

There are several results in $\real^3$~\cite{bereg2008simplifying}.  A curve $\sigma$ with the minimum $|\sigma|$ such that  $d_{\text{dF}}(\sigma,\tau) \leq \delta$ can be computed in $O(n\log n)$ time; if the vertices of $\sigma$ must be a subset of those of $\tau$, the computation time increases to $O(n^2)$.  On the other hand, if $k$ is given instead of $\delta$, a curve $\sigma$ such that $|\sigma|=k$ and $d_{\text{dF}}(\sigma,\tau)$ is minimized can be computed in $O(kn\log n\log (n/k))$ time; if the vertices of $\sigma$ must a subset of those of $\tau$, the computation time increases to $O(n^3)$.

In $\real^d$, a streaming $k$-simplification algorithm was proposed in~\cite{driemel2019sublinear} with guarantees on $d_{\text{dF}}(\sigma,\tau)$, where $\tau$ is the curve seen in the stream so far.  It was shown that for any integer $k \geq 2$ and any $\eps \in (0,1)$, a curve $\sigma$ can be computed such that $|\sigma| \leq k$ and $d_{\text{dF}}(\sigma,\tau) \leq 8 \cdot \min\{d_{\text{dF}}(\sigma',\tau) : |\sigma'| \leq k\}$.  The working storage is $O(kd)$.  

Improved results have been obtained subsequently~\cite{FF2023}.  For the streaming $\delta$-simplification problem, there is a streaming algorithm in $\real^d$ that, for any $\gamma > 1$, computes a curve $\sigma$ such that $d_{\text{dF}}(\sigma,\tau) \leq \delta$ and $|\sigma| \leq \kappa(\tau,\delta/\gamma)$.  This algorithm relies on a streaming algorithm for maintaining a $\gamma$-approximate minimum closing ball (MEB) of points in $\real^d$; the processing time per vertex is asymptotically the same as the $\gamma$-approximate MEB streaming algorithm.  There are two streaming $k$-simplification algorithms in~\cite{FF2023}.  The first one uses $O(\frac{1}{\eps}kd\log\frac{1}{\eps}) + O(\eps)^{-(d+1)/2}\log^2 \frac{1}{\eps}$ working storage and maintains a curve $\sigma$ such that $|\sigma| \leq k$ and $d_{\text{dF}}(\sigma,\tau) \leq (1+\eps) \cdot \min\{d_{\text{dF}}(\sigma',\tau) : |\sigma'| \leq k\}$.  The second algorithm uses $O(\frac{1}{\varepsilon}kd\log \frac{1}{\varepsilon})$ working storage and maintains a curve $\sigma$ such that $|\sigma| \leq k$ and $d_{\text{dF}}(\sigma,\tau) \leq (1.22+\eps) \cdot \min\{d_{\text{dF}}(\sigma',\tau) : |\sigma'| \leq k\}$.

\subsection{Our results} 

We present streaming algorithms for the $\delta$-simplification and $k$-simplification problems in $\real^d$ for $d \geq 2$ under the Fr\'{e}chet distance.  Let $\tau = (v_1,v_2,\ldots)$ be the input stream.

\vspace{5pt}

\noindent {\bf Streaming $\delta$-simplification.}   Our algorithm outputs vertices occasionally and maintains some vertices in the working storage.  The simplified curve $\sigma$ for the prefix $\tau[v_1,v_i]$ in the stream so far consists of vertices that have been output and vertices in the working storage.  Vertices that have been output cannot be modified, so they form a prefix of all simplified curves to be produced in the future for the rest of the stream.

For any $\eps \in (0,1)$ and any $\delta > 0$, our algorithm produces a curve $\sigma$ for the prefix $\tau[v_1,v_i]$ in the stream so far such that $d_F(\sigma,\tau[v_1,v_i])\le (1+\eps)\delta$ and $|\sigma|\le 2\kappa(\tau[v_1,v_i], \delta)-2$. 
Let $\alpha = 2(d-1){\lfloor d/2 \rfloor}^2 + d$.
The working storage is $O(\varepsilon^{-\alpha})$.  Each vertex in the stream is processed in $O(\eps^{-\alpha}\log\frac{1}{\eps})$ time for $d \in \{2,3\}$ and $O(\eps^{-\alpha})$ time for $d \geq 4$.  If our algorithm is used in the static case, the running time on a curve of size $n$ is $O(\eps^{-\alpha}n\log\frac{1}{\eps})$. Ignoring polynomial factors in $1/\eps$, 
our time bound is a factor $n$ smaller than the $O(\eps^{2-2d}n^2\log n \log\log n)$ running time of the best static algorithm that achieves the same error and size bounds~\cite{van2019global}. 

Intuitively, our streaming algorithm finds a line segment that stabs as many balls as possible that are centered at consecutive vertices of $\tau$ with radii $(1+O(\eps))\delta$. If such a line segment ceases to exist upon the arrival of a new vertex, we start a new line segment.  Concatenating these segments forms the simplified curve.  For efficiency, we approximate each ball by covering it with grid cells of $O(\eps\delta)$ width.  In~\cite{van2019global}, vertex balls and their covers by grid cells of $O(\eps\delta)$ width are also used.  Connections between all pairs of balls are computed to create a graph such that the simplified curve corresponds to a shortest path in the graph.  In contrast, we maintain some geometric structures so that we can read off the next line segment from them.  We prove an $O(\eps^{-\alpha})$ bound on the total size of these structures, which yields the desired working storage and processing time per vertex.

\vspace{4pt}

\noindent {\bf Streaming $\pmb{k}$-simplification.}  A memory budget may cap the size of the simplified curve, motivating the streaming $k$-simplification problem.  For example, in some wildlife tracking applications, the location-acquisition device is not readily accessible after deployment, and data is only offloaded from it after an extended period~\cite{LZSSKJ2015}. 

Our algorithm maintains the simplified curve $\sigma$ in the working storage for the prefix $\tau[v_1,v_i]$ in the stream so far.
For any $k \geq 2$ and any $\varepsilon \in(0,\frac{1}{17})$, our algorithm guarantees that $|\sigma| \leq 2k-2$ and $d_F(\sigma,\tau[v_1,v_i])\le (1+\varepsilon) \cdot \min\{d_F(\sigma',\tau[v_1,v_i]) : |\sigma'| \leq k\}$.  
Recall that $\alpha = 2(d-1){\lfloor d/2 \rfloor}^2 + d$.  
The working storage is  $O((k\varepsilon^{-1}+\varepsilon^{-(\alpha+1)})\log \frac{1}{\varepsilon})$.  Each vertex in the stream is processed in $O(k\varepsilon^{-(\alpha+1)}\log^2\frac{1}{\varepsilon})$ time for $d \in \{2,3\}$ and  $O(k\varepsilon^{-(\alpha+1)}\log\frac{1}{\varepsilon})$ time for $d \geq 4$.  We first simplify as in the case of $\delta$-simplification.  When the simplified curve reaches a size of $2k-1$, the key idea is to run our $\delta$-simplification algorithm with a suitable error tolerance to bring its size down to $2k-2$. 

There is only one prior streaming $k$-simplification algorithm~\cite{abam2007streaming,abamDCG}. It works in $\real^2$ and requires the vertices of the simplified curve $\sigma$ to be a subset of those of $\tau$.  It uses $O(k^2 + k\eps^{-1/2})$ working storage, processes each vertex in $O(k\log\frac{1}{\eps})$ amortized time, and guarantees that $|\sigma| \leq 2k$ and $d_F(\sigma,\tau) \leq (4\sqrt{2}+\eps)\cdot \mathrm{optimum}$ under the requirement that the vertices of $\sigma$ must be a subset of those of $\tau$.  For $d=2$, our algorithm uses $O((k\eps^{-1} + \eps^{-5})\log\frac{1}{\eps})$ working storage, processes each vertex in $O(k\eps^{-5}\log^2\frac{1}{\eps})$ worst-case time, and guarantees that $|\sigma| \leq 2k-2$ and $d_F(\sigma,\tau) \leq (1+\eps)\cdot \mathrm{optimum}$.  Ignoring the restriction on the vertices of $\sigma$, we offer a better approximation ratio for $d_F(\sigma,\tau)$.  Although our vertex processing time bound is larger, it is worst-case instead of amortized.  The comparison between working storage depends on the relative magnitudes of $k$ and $1/\eps$.

\section{Streaming $\pmb{\delta}$-simplification}

Given a subset $X \subseteq \real^d$, we use $\conv(X)$ to denote the convex hull of $X$.  For any point $x \in \real^d$, let $B_x$ denote the $d$-ball centered at $x$ with radius $\delta$. 
Consider the infinite $d$-dimenisonal grid with $x$ as a grid vertex and side length $\eps\delta/(2\sqrt{d})$.  Let $G_x$ be the subset of grid cells that intersect the ball centered at $x$ with radius $(1+\eps/2)\delta$.  We say that an oriented line segment $s$ \emph{stabs the objects $O_1,\ldots, O_m$ in order} if there exist points $x_i \in s \cap O_i$ for $i \in [m]$ such that $x_1,\ldots,x_m$ appear in this order along $s$~\cite{guibas1993approximating}.

\subsection{Algorithm}

The high-level idea is to find the longest sequence $v_1,v_2,\ldots,v_i$ of vertices such that there is a segment $s_1$ that stabs $B_{v_a}$ for $a \in [i]$ in order.  Then, restart to find the longest sequence $v_{i+1},\ldots,v_j$ such that there is a segment $s_2$ that stabs $B_{v_a}$ for $a \in [i+1,j]$ in order.  Repeating this way gives a sequence of segments $s_1, s_2, \ldots$.  Concatenating these segments gives the simplified curve.  The problem is the large working storage: a long vertex sequence may be needed to determine a line segment.  We use several ideas to overcome this problem.  

First, we approximate $B_{v_a}$ by $\conv(G_{v_a})$ and find a line segment that stabs the longest sequence $\conv(G_{v_1}),\ldots, \conv(G_{v_i})$ in order.  We will show that it suffices for this line segment to start from the set $P$ of grid points in $G_{v_{1}}$.  So $|P| = O(\eps^{-d})$.

Second, as the vertices arrive in the stream, for each point $p \in P$, we construct a structure $S_a[p]$ inductively as follows:
\begin{quote}
\begin{itemize}
\item Set $S_1[p] = \{p\}$ for all $p \in P$.  We write $S_1[p] = p$ for convenience.

\vspace{2pt}

\item For all $a \geq 1$, set $S_{a+1}[p] = \conv(G_{v_{a+1}}) \cap F(S_a[p],p)$ for all $p \in P$, where $F(S_a[p],p) = \{y \in \real^d :  py \cap S_a[p] \not= \emptyset\}$.

\end{itemize}
\end{quote}
If $p \in S_a[p]$, then $F(S_a[p],p) = \real^d$ and hence $S_{a+1}[p] = \conv(G_{v_{a+1}})$.  If $S_a[p] = \emptyset$, then $F(S_a[p],p) = \emptyset$ and hence $S_{a+1}[p] = \emptyset$.   Otherwise, by viewing $p$ as a light source, $F(S_a[p],p)$ is the unbounded convex polytope consisting of $S_a[p]$ and the non-illuminated subset of $\real^d$.  

We will show that for $a \geq 2$, $S_a[p]$ is equal to $\{x \in \conv(G_{v_a}) : \text{$px$ stabs $\conv(G_{v_1})$}$, $\ldots$, $\conv(G_{v_{a}}) \text{ in order}\}$.  We will also show that for each $p \in P$, $S_a[p]$ and $F(S_a[p],p)$ have $\mathrm{poly}(1/\eps)$ complexities, and they can be constructed in $\mathrm{poly}(1/\eps)$ time.  The array $S_a$ is deleted after computing $S_{a+1}$, which keeps the working storage small.

We pause when we encounter a vertex $v_{i+1}$ such that $S_{i+1}[p] = \emptyset$ for all $p \in P$.  We choose any vertex $q$ of any non-empty $S_i[p]$ (which must exist) and output the segment $pq$.  Since $pq$ stabs $\conv(G_{v_1}), \ldots, \conv(G_{v_i})$ in order, we can show that $d_F(pq,\tau[v_1,v_i]) \leq (1+\eps)\delta$.  All is well except that $S_i$ has been deleted after computing $S_{i+1}$.  A fix is that after processing a vertex $v_a$, we store a segment $pq$ in a buffer $\tilde{\sigma}$ for an arbitrary vertex $q$ of an arbitrary non-empty $S_a[p]$.  We output the content of $\tilde{\sigma}$ after discovering that $S_{i+1}[p] = \emptyset$ for all $p \in P$. 

Afterward, we reset $P$ as the set of grid points in $G_{v_{i+1}}$ and restart the above processing from $v_{i+1}$.  That is, reset $S_{i+1}[p] = p$ for all $p \in P$.
One of the points in $P$ will be the start of the next segment that will be output in the future.  The start of this next segment may be far from $\conv(G_{v_i})$, which contains the end of $\tilde{\sigma}$ that was just output.  For $a = i+1, i+2, \ldots$, we set $S_{a+1}[p] = \conv(G_{v_{a+1}}) \cap F(S_a[p],p)$ for all $p \in P$ until $S_a[p]$ becomes empty for all $p \in P$ again.  In summary, the simplified curve $\sigma$ is the concatenation of the output segments $s_1, s_2, \ldots$, i.e., the end of $s_j$ is connected to the start of $s_{j+1}$.  Figure~\ref{fg:simplify} illustrates a few steps of this algorithm.  The procedure {\sc Simplify}$(\eps,\delta)$ in Algorithm~\ref{alg:simplify} gives the details.

\begin{figure}[h]
    \centerline{\includegraphics[scale=0.7]{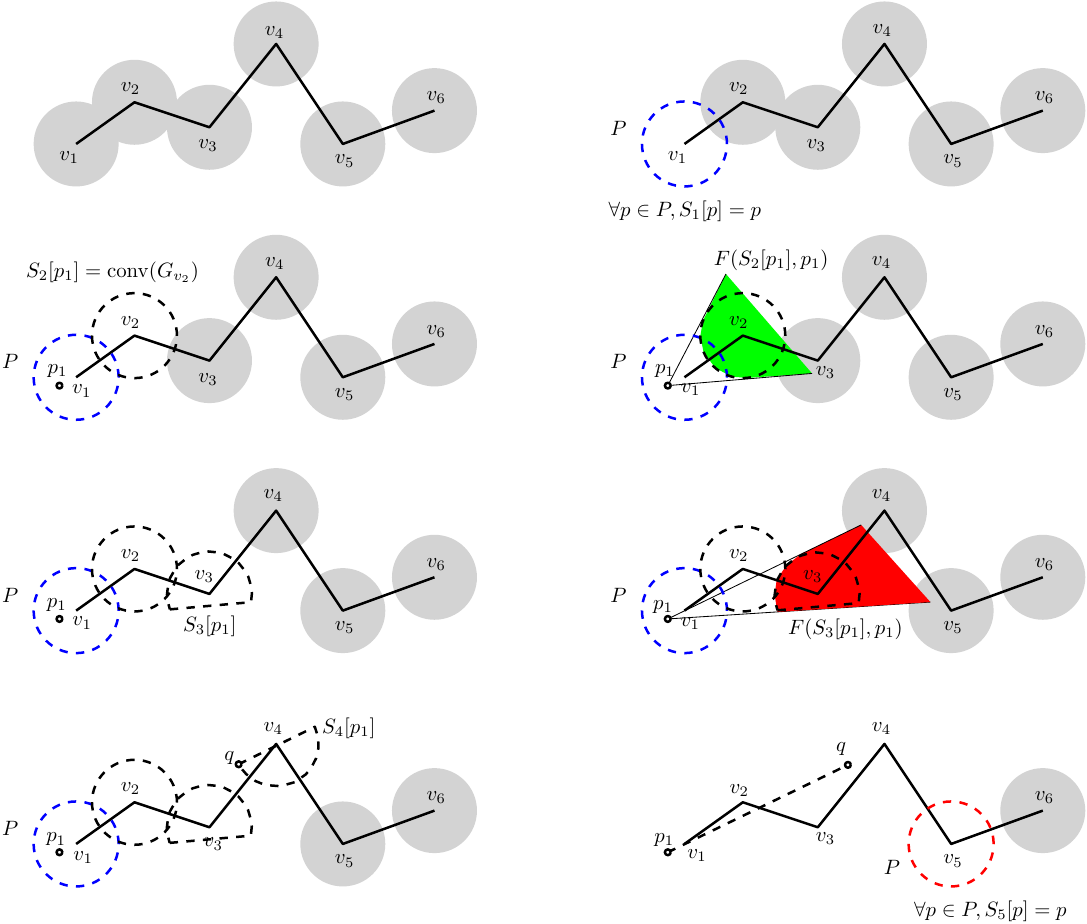}}
    \caption{The shaded balls denote the $\conv(G_{v_a})$'s.  The vertices arrive in order (from left to right and top to bottom).  The first $P$ consists of the grid points inside the blue dashed ball.  Let $p_1$ be a grid point in $P$.  Note that $F(S_1[p_1],p_1) = \real^d$.  So $S_2[p_1] = \conv(G_{v_2})$ and $F(S_2[p_1],p_1)$ is the green unbounded convex region that is bounded by the tangents from $p_1$ to $S_2[p_1]$.  $S_3[p_1] = \conv(G_{v_3}) \cap F(S_2[p_1],p_1)$ is denoted by the dashed clipped ball around $v_3$.  Similarly, $F(S_3[p_1],p_1)$ is the red unbounded convex region that is bounded by the tangents from $p_1$ to $S_3[p_1]$.  $S_4[p_1]$ is denoted by the dashed clipped ball around $v_4$. In this example, $S_5[p] = \emptyset$ for all $p \in P$, so we output $\tilde{\sigma}$ which may be equal to $(p_1,q)$ for a vertex $q$ of $S_4[p_1]$.  We reset $P$ as the set of grid points in $G_{v_5}$ (inside the red dashed ball) and reset $S_5[p] = p$ for all $p \in P$.}
    \label{fg:simplify}
\end{figure}

\begin{algorithm}[h]
\caption{{\sc Simplify}}
\label{alg:simplify}
{\bf Input:} A stream $\tau = (v_1,v_2,\ldots)$, $\eps$, and $\delta$.

{\bf Output:} The output curve $\sigma$ consists of the sequence of vertices output so far and the vertices in the buffer $\tilde{\sigma}$.  After processing the last vertex $v_i$ at the end of the stream, {\sc Simpify} returns $(P,S_i)$, which will be useful for the $k$-simplification problem.
\begin{algorithmic}[1]
\Function{{\sc Simplify}}{$\eps,\delta$}
\State {read $v_1$ from the data stream; $i \gets 1$; $\mathit{init} \gets \mathrm{true}$; $\tilde{\sigma} \gets \emptyset$}
\While {{\bf true}}
    \If {{\it init} = true} 
        \State {$\mathit{init} \gets \mathrm{false}$}
        \State {output the vertices in $\tilde{\sigma}$} 
        \Comment {output the line segment in $\tilde{\sigma}$}
        \State {$\tilde{\sigma} \gets (v_i)$}  \Comment{view $v_i$ as a degenerate line segment}
        \State {$P \gets$ the grid points of $G_{v_i}$}
        \State {$S_i[p] \gets p$ for all $p \in P$}
    \Else
        \State {choose a vertex $q$ of some non-empty $S_i[p]$}
        \State {$\tilde{\sigma} \gets (p,q)$} \Comment{update the last edge of the simplified curve}
    \EndIf
    \If {end of data stream}
        \hspace*{45pt}\State {output $\tilde{\sigma}$} \\
        \hspace*{49pt}\Return{$(P,S_i)$}  \Comment{will be useful for $k$-simplification}
    \Else
        \State {read $v_{i+1}$ from the data stream}
    \EndIf
    \State{$S_{i+1}[p] \gets \conv(G_{v_{i+1}}) \cap F(S_i[p],p)$ for all $p \in P$}
    \State {delete $S_i$}
    \State{$\mathit{init} \gets \mathrm{true}$ if $S_{i+1}[p] = \emptyset$ for all $p \in P$}
    \State {$i \gets i + 1$}
\EndWhile
\EndFunction
\end{algorithmic}
\end{algorithm}

\subsection{Analysis}

Given a polytope $Q \subset \real^d$, let $|Q|$ denote its complexity, i.e., the total number of its faces of all dimensions.

\begin{lemma}
\label{lem:stab}
Let $P$ be the set of grid points in $G_{v_i}$.  Assume that $S_i[p]=p$ for all $p \in P$.  Take any $p \in P$ and any index $j \geq i+1$ such that $S_a[p] \not= \emptyset$ for all $a \in [i,j-1]$.
\begin{enumerate}[{\em (i)}]
\item $S_j[p] = \{x \in \conv(G_{v_j}) : \text{$px$ stabs $\conv(G_{v_i}), \ldots, \conv(G_{v_j})$ in order}\}$.
\item $S_j[p]$ is a convex polytope that can be computed in $O\bigl(\bigl|S_{j-1}[p]\bigr| + N\log N + N^{\lfloor d/2 \rfloor}\bigr)$ time, where $N$ is any upper bound on the number of support hyperplanes of $F(S_{j-1}[p],p)$ such that $N = \Omega(\eps^{(1-d)\lfloor d/2 \rfloor})$.
\end{enumerate}
\end{lemma}
\begin{proof}
When $j = i+1$, (i) is true as $S_{i+1}[p] = \conv(G_{v_{i+1}})$.  Assume inductively that (i) holds for $j-1$ for some $j \geq i+2$.  For any $x \in F(S_{j-1}[p],p)$, $px$ intersects $S_{j-1}[p]$ by definition, which implies that $px$ stabs $\conv(G_{v_i}),\ldots,\conv(G_{v_{j-1}})$ in order by induction assumption.  Therefore, for every $x \in \conv(G_{v_j}) \cap F(S_{j-1}[p],p)$, $px$ stabs $\conv(G_{v_i}), \ldots, \conv(G_{v_j})$ in order.  Conversely, for every $x \in \conv(G_{v_j})$, if $px$ stabs $\conv(G_{v_i}),\ldots,\conv(G_{v_j})$ in order, then $px$ contains a point in $S_{j-1}[p]$ by induction assumption, which implies that $x \in F(S_{j-1}[p],p)$.

Consider (ii).  If $p \in S_{j-1}[p]$, then $F(S_{j-1}[p],p) = \real^d$ and $S_j[p] = \conv(G_{v_j})$.  
This case can be handled in $O\bigl(\bigl|S_{j-1}[p]\bigr|\bigr)$ time.  If $p \not\in S_{j-1}[p]$,
then $F(S_{j-1}[p],p)$ has two types of support hyperplanes.  The first type consists of support hyperplanes of $S_{j-1}[p]$ that separate $p$ from $S_{j-1}[p]$.  The second type consists of hyperplanes that pass through $p$ and are tangent to $S_{j-1}[p]$ at some $(d-2)$-dimensional faces of $S_{j-1}[p]$. 
%
%
%
%
Since the side length of $G_{v_j}$ is $\eps\delta/(2\sqrt{d})$, the number of cells in $G_{v_j}$ that intersect the boundary of $B_{v_j}$ is $O(\eps^{1-d})$.  Hence, $\conv(G_{v_j})$ has $O(\eps^{(1-d)\lfloor d/2 \rfloor})$ complexity which is $O(N)$, and $S_j[p]$ can be computed in $O(N\log N + N^{\lfloor d/2 \rfloor})$ time~\cite{C1993}. 
\end{proof}

Next, we prove the solution quality guarantees.

\begin{lemma}
\label{lem:quality}
Let $\tau[v_1,v_i]$ be the prefix in the stream so far.  The simplified curve $\sigma$ satisfies the properties that $d_F(\sigma,\tau[v_1,v_i]) \leq (1+\eps)\delta$ and $|\sigma| \leq 2\kappa(\tau[v_1,v_i],\delta)-2$.
\end{lemma}
\begin{proof}
Let $\sigma$ be the concatenation of line segments $s_1, s_2, \ldots$.  That is, the end of $s_j$ is connected to the start of $s_{j+1}$ for $j \geq 1$ to produce $\sigma$.

The start of $s_1$ is 
a grid point $p_1$ of $G_{v_1}$, the end of $s_1$ is a vertex $p_b$ of $S_b[p_1] \subseteq \conv(G_{v_b})$ for some $b > 1$, and the start of $s_2$ is a grid point $p_{b+1}$ in $G_{v_{b+1}}$.  By Lemma~\ref{lem:stab}(i), $s_1$ stabs $\conv(G_{v_1}), \ldots, \conv(G_{v_b})$ in order.  Let $p_2,\ldots,p_{b-1}$ be some points on $s_1$ in this order such that $p_a \in s_1 \cap \conv(G_{v_a})$ for $a \in [2,b-1]$. Then, $d(v_a,p_a) \leq (1+\eps)\delta$ for $a \in [b]$.  So we can match $v_av_{a+1}$ to $p_ap_{a+1}$ for $a \in [b-1]$ by linear interpolation within a distance of $(1+\eps)\delta$.  Similarly, we can match $v_bv_{b+1}$ to the segment $p_bp_{b+1}$ between $s_1$ and $s_2$ by linear interpolation within a distance of $(1+\eps)\delta$.  Continuing this way shows that $d_F(\sigma,\tau[v_1,v_i]) \leq (1+\eps)\delta$.

\begin{figure}
\centerline{\includegraphics[scale=0.6]{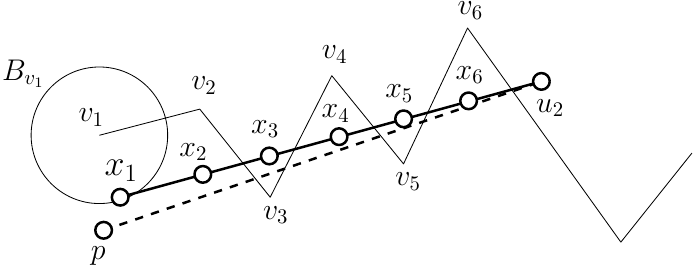}}
\caption{The points $x_i$ for $i \in [6]$ are matched to $v_i$ for $i \in [6]$ by the Fr\'{e}chet matching.}
\label{fg:quality}
\end{figure}

Let $k = \kappa(\tau[v_1,v_i],\delta)$.  Let $\gamma = (u_1,u_2,\ldots,u_k)$ be a polygonal curve of size $k$ such that $d_F(\gamma,\tau[v_1,v_i]) \leq \delta$.  For $a \in [i]$, let $x_a$ be a point on $\gamma$ mapped to $v_a$ by a Fr\'{e}chet matching between $\gamma$ and $\tau[v_1,v_i]$.  We can assume that $x_1 = u_1$ and $x_i = u_k$.  Let $b = \max\{a \in [i]: x_a \in u_1u_2\}$.  So $x_1u_2$ stabs balls of radii $\delta$ centered at $v_1,v_2,\ldots,v_b$ in order.  Refer to Figure~\ref{fg:quality}.  There is a grid point $p$ of $G_{v_1}$ within a distance $\eps\delta/2$ from $x_1$.  By a linear interpolation between $x_1u_2$ and $pu_2$, for $i \in [b]$, $x_i$ is mapped to a point on $pu_2$ at distance no more than $(1+\eps/2)\delta$ from $v_i$. That is, $pu_2$ stabs the balls centered at $v_1, v_2, \ldots, v_b$ with radii $(1+\eps/2)\delta$ in order.  It follows that $pu_2$ stabs $\conv(G_{v_1}),\ldots,\conv(G_{v_b})$ in order, which implies that $S_a[p] \not= \emptyset$ for $a \in [b]$ by Lemma~\ref{lem:stab}(i).  Therefore, {\sc Simplify} outputs the first segment when it encounters $v_{c+1}$ for some $c \geq b$.  We charge the first output segment to $u_1u_2$.  The vertex $v_{c+1}$ is matched to the point $x_{c+1}$ that lies on $u_ju_{j+1}$ for some $j \geq 2$.  We repeat the above argument to $x_{c+1}u_{j+1}$.  This way, every output segment is charged to a unique edge of $\gamma$, so there are at most $k-1$ output segments.
Concatenating them uses at most  $k-2$ more edges. Hence, $|\sigma| \leq 2k-2$.
\end{proof}

We are now ready to prove the  performance guarantees offered by {\sc Simplify}.

\begin{theorem}\label{thm: delta_main}
Let $\tau$ be a polygonal curve in $\mathbb{R}^d$ that arrives in a data stream. Let $\alpha  = 2(d-1)\lfloor d/2 \rfloor^2 + d$.
For every $\delta > 0$ and every $\varepsilon \in (0,1)$, the output curve $\sigma$ of {\sc Simplify}$(\eps,\delta)$ satisfies $d_F(\sigma,\tau)\le (1+\varepsilon)\delta$ and $|\sigma|\le 2\kappa(\tau, \delta)-2$. 
The working storage is $O(\varepsilon^{-\alpha})$.  Each vertex is processed in $O(\varepsilon^{-\alpha}\log\frac{1}{\eps})$ time for $d \in \{2,3\}$ and $O(\varepsilon^{-\alpha})$ time for $d \geq 4$.
\end{theorem}
\begin{proof}
The guarantees on $\sigma$ follow from Lemma~\ref{lem:quality}.  We will bound the number of support hyperplanes and complexities of $S_j[p]$ and $F(S_{j-1}[p],p)$.  The working storage and processing time then follow from Lemma~\ref{lem:stab}(ii) and the fact that $|P| = O(\eps^{-d})$.

A \emph{bounding halfspace} of a convex polytope $O$ is a halfspace that contains $O$ and is bounded by the support hyperplane of a facet of $O$.  We count the bounding halfspaces of $S_j[p]$.  Some are bounding halfspaces of $F(S_{j-1}[p],p)$ whose boundaries pass through $p$ and are tangent to $S_{j-1}[p]$.  We call them the \emph{pivot} bounding halfspaces.  The remaining ones are bounding halfspaces of $\conv(G_{v_a})$ for some $a \leq j$.  We call these \emph{non-pivot} bounding halfspaces.   

Since $\conv(G_{v_a})$'s are translates of each other, their facets share a set $V$ of unit outward normals.  We have $|V| \leq |\conv(G_{v_a})| = O(\eps^{(1-d)\lfloor d/2 \rfloor})$.  No two facets of $S_j[p]$ have the same vector in $V$ because the two corresponding bounding halfspaces would be identical or nested---neither is possible.  This feature helps us to prevent a cascading growth in the complexities of $S_j[p]$ in the inductive construction.
It also implies that there are $O(\eps^{(1-d)\lfloor d/2 \rfloor})$ non-pivot  bounding halfspaces of $S_j[p]$.

Take a pivot bounding halfspace $h$ of $S_j[p]$.  Suppose that $h$ was introduced for the first time as a bounding halfspace of $S_i[p]$ for some $i \leq j$.  The boundary of $h$, denoted by $\partial h$, passes through $p$ and is tangent to $S_{i-1}[p]$ at a $(d-2)$-dimensional face $f$ of $S_{i-1}[p]$.  Let $g_1$ and $g_2$ be the bounding halfspaces of $S_{i-1}[p]$ whose boundaries contain the two facets of $S_{i-1}[p]$ incident to $f$.  
We make three observations.  First, $F(S_{i-1}[p],p)$ lies inside the cone $C_i$ of rays that shoot from $p$ towards $S_{i-1}[p]$.  Second, $h$ is a bounding halfspace of $C_i$.  Third, $S_a[p] \subseteq C_i$ for all $a \in [i,j]$.  

We claim that neither $\partial g_1$ nor $\partial g_2$ passes through $p$.  Neither $g_1$ nor $g_2$ is equal to $h$ because $h$ was not introduced before $S_i[p]$.  If both $\partial g_1$ and $\partial g_2$ pass through $p$, then $f$ spans a $(d-2)$-dimensional affine subspace $\ell$ that passes through $p$.  
But then $\partial h$ meets $C_i$ at $\ell$ only, which is a contradiction because $\partial h$ should support a facet of $S_j[p]$.  If either $\partial g_1$ or $\partial g_2$ passes through $p$, say $\partial g_2$, then the affine subspace spanned by $f$ does not pass through $p$.  But then $p$ and $f$ span a unique hyperplane, giving the contradiction that $g_2 = h$.  

By our claim, $g_1$ and $g_2$ are non-pivot bounding halfspaces of $S_{i-1}[p]$.  We charge $h$ to $(\phi_1,\phi_2)$ or $(\phi_2,\phi_1)$, where $\phi_1$ and $\phi_2$ are the outward unit normals of $g_1$ and $g_2$, respectively.  The pair $(\phi_1,\phi_2)$ is charged if a ray that shoots from $p$ to an interior point of $g_1 \cap g_2$ arbitrarily near $\partial g_1 \cap \partial g_2$ hits $\partial g_1$ before $\partial g_2$.  Otherwise, the pair $(\phi_2,\phi_1)$ is charged.  Suppose that $h$ is charged to $(\phi_1,\phi_2)$.  We argue that $(\phi_1,\phi_2)$ is not charged again at some $(d-2)$-dimensional face of $S_a[p]$ for all $a \in [i,j-1]$ due to another pivot bounding halfspace of $S_j[p]$. 

 Assume to the contrary that $(\phi_1,\phi_2)$ is charged again at some $(d-2)$-dimensional face of $S_b[p]$ for some $b \in [i,j-1]$ due to a pivot bounding halfspace $h'$ of $S_j[p]$.  
%
%
Either the dimension of $\partial h' \cap C_i$ is zero or one, or the dimension of $\partial h' \cap C_i$ is two.  In the first case, $\partial h'$ cannot support any facet of $S_j[p]$, a contradiction.  In the second case, since $h'$ is charged to the same ordered pair $(\phi_1,\phi_2)$, $\partial h'$ must separate $S_b[p]$ from $\partial h$.  Figure~\ref{fg:simplify-2} illustrates the configuration. Therefore, the cone of rays $C_b$ that shoot from $p$ towards $S_b[p]$ meets $\partial h$ only at $p$.  However, $S_j[p] \subseteq C_b$, which means that $\partial h$ cannot support any facet of $S_j[p]$, a contradiction.

\begin{figure}
	\centerline{\includegraphics[scale=0.5]{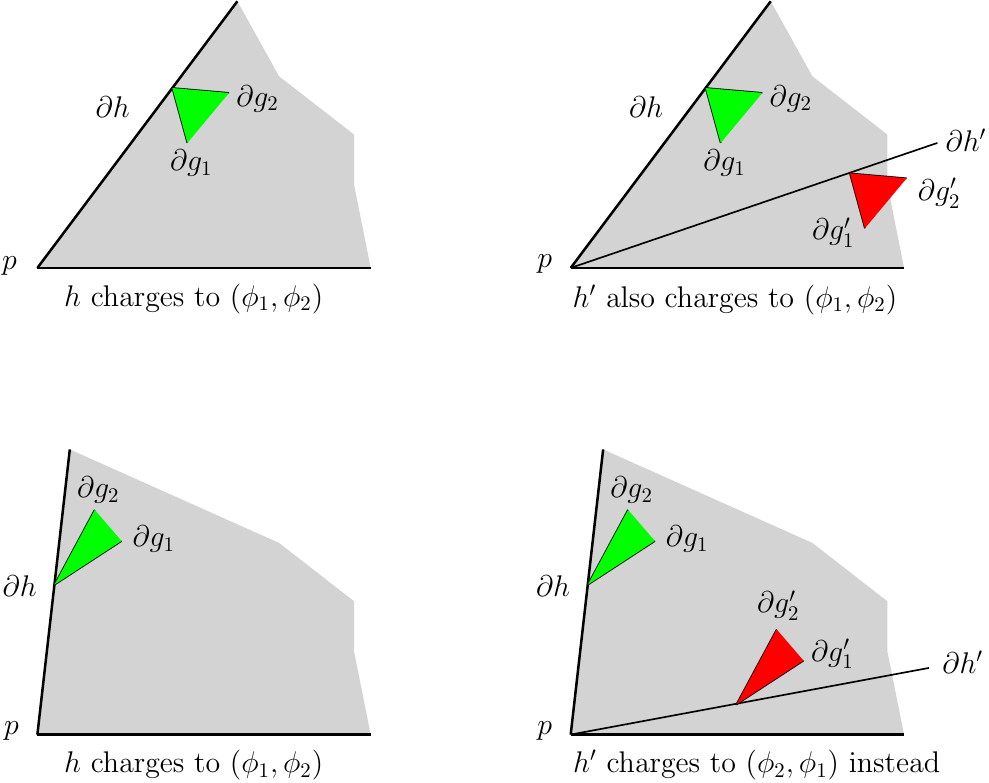}}
	\caption{The shaded cones denote $C_i$.  Let $g'_1$ and $g'_2$ be the bounding halfspaces of $S_b[p]$ such that $\partial h'$ is spanned by $p$ and $\partial g'_1 \cap \partial g'_2$.  Without loss of generality, assume that $g_i$ and $g'_i$ have the same unit outward normal $\phi_i$ for $i \in \{1,2\}$.  A local neighborhood of $g_1 \cap g_2$ is shown in green, and a local neighborhood of $g'_1 \cap g'_2$ is shown in red.  In the first row, both $h$ and $h'$ are charged to $(\phi_1,\phi_2)$, and $\partial h'$ separates $g'_1 \cap g'_2 \cap C_i$ from $\partial h$.  Note that $S_b[p] \subseteq g'_1 \cap g'_2 \cap C_i$.  In the second row, the alternative configuration is shown in which $\partial h'$ does not separate $g'_1 \cap g'_2 \cap C_i$ from $\partial h$.  In this case, a ray that shoots from $p$ to any interior point of $g'_1 \cap g'_2$ arbitrarily close to $\partial g'_1 \cap \partial g'_2$ must intersect $\partial g'_2$ before $\partial g'_1$. Therefore, $h'$ is charged to $(\phi_2,\phi_1)$ instead of $(\phi_1,\phi_2)$.}
	\label{fg:simplify-2}
\end{figure}

In all, $S_j[p]$ has $O(\eps^{(1-d)\lfloor d/2 \rfloor})$ non-pivot bounding halfspaces and at most $|V|(|V|-1) = O(\eps^{2(1-d)\lfloor d/2 \rfloor})$ pivot bounding halfspaces.  The same bounds hold for $S_{j-1}[p]$.  If a support hyperplane $L$ of $F(S_{j-1}[p],p)$ is not a support hyperplane of $S_{j-1}[p]$, then $L$ passes through $p$ and is tangent to $S_{j-1}[p]$ at some $(d-2)$-dimensional face $f$.  Since $L$ is different from the support hyperplanes of the two facets of $S_{j-1}[p]$ incident to $f$, these two hyperplanes do not pass through $p$.  So the two facets incident to $f$ induce some non-pivot bounding halfspaces of $S_{j-1}[p]$, implying that $F(S_{j-1}[p],p)$ has $O(\eps^{2(1-d)\lfloor d/2 \rfloor})$ bounding halfspaces.  
It follows that $\bigl|S_j[p]\bigr|$ and $\bigl|F(S_{j-1}[p],p)\bigr|$ are $O(\eps^{2(1-d){\lfloor d/2 \rfloor}^2})$.
\end{proof}

\section{Streaming algorithm for the $\pmb k$-simplification problem}\label{sec: k-simplification}
	
\subsection{Algorithm}

Let $\sigma_i^*$ be the optimal solution for the $k$-simplification problem for $\tau[v_1,v_i]$.  If we knew $\delta_i^* = d_F(\sigma_i^*,\tau[v_1,v_i])$, we could call {\sc Simplify}$(\eps,\delta_i^*)$ to obtain an approximate solution for the $k$-simplification problem. Therefore, we try to estimate $\delta_i^{*}$ on the fly.  

We maintain an output simplified curve $\sigma_i$ in the working storage after processing every vertex $v_i$ in the stream.  The first simplified curve $\sigma_{2k-1}$ is obtained by calling {\sc Simplify}, via an invocation of a new procedure {\sc Compress}, on $\tau[v_1,v_{2k-1}]$ with $\delta$ equal to some value $\delta_{2k-1}$. Given a curve of size $2k-1$, {\sc Compress} simplifies it to a curve of size at most $2k-2$.  

For $i \geq 2k$, when a vertex $v_i$ arrives in the stream, either we update the last edge of $\sigma_{i-1}$ to produce $\sigma_i$ as in {\sc Simplify}, or we start a new segment (initially just the vertex $v_i$) as in {\sc Simplify}.  Suppose that we start a new segment.  If $|\sigma_{i-1}| < 2k-2$, we can append $v_i$ to $\sigma_{i-1}$ to form $\sigma_i$.  Otherwise, $\sigma_i$ is obtained by calling {\sc Simplify}, via an invocation of {\sc Compress},  on the concatenation $\sigma_{i-1} \circ (v_i)$ with $\delta$ equal to some value $\delta_{i}$ to be specified later.  We will prove in Lemmas~\ref{lem: delta-min}--\ref{lem: deal_with_alpha} that $\delta_i \leq (1+O(\eps))\delta^*_i$.  Afterward, we repeat the above to process the next vertex in the stream.  The above processing is elaborated in {\sc Reduce} in
Algorithm~\ref{alg: k_simplification}.  

We define a curve $\tau_i$ for every $i \geq 2k-1$ in the comments in lines~\ref{alg:tau-1}, \ref{alg:tau-2}, and~\ref{alg:3} of {\sc  Reduce}.  These curves facilitate the analysis, but they are not maintained by the algorithm as variables.  For each $i$, $\tau_i$ is the curve from which the simplification $\sigma_i$ is computed.  Therefore, $\tau_{2k-1} = \tau[v_1,v_{2k-1}]$ (line~\ref{alg:tau-1}).  Consider any $i \geq 2k$.   If we call {\sc Compress} to simplify $\sigma_{i-1} \circ (v_i)$ to a curve $\sigma_i$ of size at most $2k-2$, we set $\tau_i = \sigma_{i-1} \circ (v_i)$ (line~\ref{alg:3}).  Otherwise, $\sigma_i$ is obtained by simplifying $\tau_{i-1} \circ (v_i)$ whether we start a new segment at $v_i$ or not, so $\tau_i = \tau_{i-1} \circ (v_i)$ (line~\ref{alg:tau-2}).

The input parameter $r$ of {\sc Reduce} needs an explanation. 
We will show how to compute a lower bound $\delta_{\min} > 0$ for $d_F(\sigma^*_{2k-1},\tau[v_1,v_{2k-1}])$. For $i \geq 2k-1$, let $r_i \geq 1$ and $s_i \geq 0$ be integers such that 
\begin{equation*}
\frac{1}{\eps^{s_i}}\delta_{\min} 
\leq 
\frac{(1+\varepsilon)^{r_i-1}}{\varepsilon^{s_i}}\delta_{\min}
\leq \delta_i^*
\leq 
\frac{(1+\varepsilon)^{r_i}}{\varepsilon^{s_i}}\delta_{\min}
\leq 
\frac{1}{\eps^{s_i+1}}\delta_{\min}.
\end{equation*}
We will show that if we can start the processing of the data stream with an initial error tolerance of $\delta_{2k-2} = \eps(1+\varepsilon)^{r_i+1}(1-4\eps)^{-1}\delta_{\min}$, then the output curve $\sigma_i$ will be within a Fr\'{e}chet distance of $\eps^{-s_i}(1+\varepsilon)^{r_i+1}(1-4\eps)^{-1}\delta_{\min}$ from $\tau[v_1,v_i]$.  By the definitions of $r_i$ and $s_i$, the error is at most $(1+O(\eps))\delta_i^*$.  The correct values of $r_i$ and $s_i$ are unknown to us.  Nevertheless, no matter what $s_i$ is, we have $1 \leq r_i \leq \lfloor \log_{1+\eps}\frac{1}{\eps}\rfloor$.  We launch $\lfloor \log_{1+\eps}\frac{1}{\eps}\rfloor=O(\frac{1}{\varepsilon}\log \frac{1}{\varepsilon})$ independent runs of {\sc Reduce} for each $r$ between 1 and $\lfloor \log_{1+\eps}\frac{1}{\eps}\rfloor$. After processing the arriving vertex $v_i$, the run with the minimum $\delta_i$ gives the right answer. This idea of using multiple independent runs to capture the right $\delta$ was proposed in~\cite{FF2023} for streaming $k$-simplification under the discrete  Fr\'{e}chet distance.  We still have to obtain $s_i$.  This is achieved by establishing a lower bound for $\delta^*_i$ (Lemma~\ref{lem:shortcut}) and computing a particular upper bound for $\delta_i^*$ in line~\ref{alg:shortcut} of {\sc Compress}.

We assume that no three consecutive vertices in the stream are collinear.  We enforce it using $O(1)$ more working storage as follows.  Remember the last two vertices $(x, y)$ in the stream, but do not feed $y$ to {\sc Reduce} yet.  Wait until the next vertex $z$ arrives.  If $x$, $y$ and $z$ are not colinear, feed $y$ to {\sc Reduce} and remember $(y,z)$ as the last two vertices.  If $x$, $y$ and $z$ are collinear, then drop $y$, make $(x,z)$ the last two vertices, and wait for the next vertex.
%

Algorithm~\ref{alg:repeatsimplify} shows the procedure {\sc Compress}. If {\sc Compress} is called when processing $v_i$, its task is to simplify $\tau_i = \sigma_{i-1} \circ (v_i)$ to a curve $\sigma_i$ of size at most $2k-2$ with error estimate $\delta_i = \eps^{-t}\delta_{i-1}$, where $t$ is some judiciously chosen positive integer.  


\begin{algorithm}
\caption{{\sc Reduce}}
\label{alg: k_simplification}
{\bf Input:} A data stream $\tau = (v_1,v_2,\ldots)$ and three parameters $r$, $\varepsilon$, and $k$. \\
{\bf Output:} Let $\tau[v_1,v_i]$ be the prefix in the stream so far.  A curve $\sigma_i$ is maintained such that $|\sigma_i| \leq 2k-2$ and $d_F(\sigma_i,\tau[v_1,v_i]) \leq (1+O(\eps))\delta_i$.
\begin{algorithmic}[1]
\Procedure{Reduce}{$r,\varepsilon,k$}
\State {read the first $2k-1$ vertices $v_1,\ldots,v_{2k-1}$ of $\tau$}
\State {$\delta_{\min} \gets \frac{1}{2}\min\bigl\{d(v_i,v_{i-1}v_{i+1}) : i \in [2,2k-2]\bigr\}$} \label{alg:lowerbound} 
\State {$\delta_{2k-2} \gets \eps(1+\varepsilon)^{r+1}(1-4\eps)^{-1}\delta_{\min}$}  \label{alg:1-start} 
\State{$\sigma_{2k-2} \gets \tau[v_1,v_{2k-2}]$}   
\Comment{$\tau_{2k-1} \gets \tau[v_1,v_{2k-1}]$}   \label{alg:tau-1}
\State {$(\sigma_{2k-1},\delta_{2k-1},P,S_{2k-1}) \gets \text{{\sc Compress}}(\tau[v_1,v_{2k-1}],\varepsilon,k,\delta_{2k-2})$} \label{alg:1-end} 
\State {$i \gets 2k$}
\While{{\bf true}}
\State {read $v_i$ from the data stream}   \label{alg:2-start}
\State{$\delta_i \gets \delta_{i-1}$}
\Comment{$\tau_i \gets \tau_{i-1} \circ (v_i)$}  \label{alg:tau-2}
\State{$S_i[p] \gets \conv(G_{v_{i}}) \cap F(S_{i-1}[p],p)$ for all $p \in P$, where $G_{v_{i}}$ is defined with \hspace*{72pt}$\delta = \delta_{i}$}

\If {$S_i[p] \not= \mathrm{null}$ for some $p \in P$} \\\Comment{update the last possibly degenerate segment as in {\sc Simplify}}
    \State{$q \gets$ any vertex of $S_i[p]$}
    \State{$\sigma_i \gets \sigma_{i-1}$ with the last segment replaced by $pq$}  
\ElsIf {$|\sigma_{i-1}| < 2k-2$}  \Comment{start a new segment as in {\sc Simplify}}
    \State{$P \gets$ the grid points of $G_{v_i}$}  
    \Comment{$G_{v_i}$ is defined with $\delta = \delta_i$}
    \State{$S_i[p] \gets p$ for all $p \in P$}
    \State{$\sigma_i \gets \sigma_{i-1} \circ (v_i)$} 
    \Comment{treat $v_i$ as a degenerate segment}
    \label{alg:2-end}
\Else  \Comment{simplify $\sigma_{i-1}\circ (v_i)$ to produce $\sigma_i$}
    \State {$(\sigma_i,\delta_i,P,S_i) \gets \text{{\sc Compress}}(\sigma_{i-1} \circ (v_i),\varepsilon,k,\delta_{i-1})$}  
    \Comment{$\tau_i = \sigma_{i-1} \circ (v_i)$} \label{alg:3}  
\EndIf
\State{delete $\sigma_{i-1}$, $\delta_{i-1}$, and $S_{i-1}$}
\State {$i \gets i +1$}
\EndWhile
\EndProcedure\end{algorithmic}
\end{algorithm}

\begin{algorithm}
\caption{{\sc Compress}}
\label{alg:repeatsimplify}
{\bf Input:} A polygonal curve $(x_1,\ldots,x_{2k-1})$ and three parameters $\varepsilon$, $k$, and $\delta$. \\
{\bf Output:} $(\zeta,\eps^{-t}\delta,P,S)$, where $\zeta$ is the output curve of calling {\sc Simplify}$(\eps,\eps^{-t}\delta)$ on $(x_1,\ldots,x_{2k-1})$ for some $t \in \mathbb{N}$, and $(P,S)$ are the output returned by {\sc Simplify} after processing $x_{2k-1}$.
\begin{algorithmic}[1]
\Function{Compress}{$(x_1,\ldots,x_{2k-1}),\varepsilon,k,\delta$}
\State {$t \gets \min\left\{i \in \mathbb{N} : i \geq 1, \eps^{-i}\delta \geq \frac{1}{2}\min\{d(x_j,x_{j-1}x_{j+1}) : j \in [2,2k-2], \, \text{$j$ is even}\}\right\}$}  \label{alg:shortcut}
\State{$\zeta \gets \text{output curve produced by calling {\sc Simplify}}(\eps,\eps^{-t}\delta)$ on $(x_1,\ldots,x_{2k-1})$} 
\State{$(P,S) \gets$ output returned by the above call of {\sc Simplify} after processing $x_{2k-1}$}
\State {\Return $(\zeta,\eps^{-t}\delta,P,S)$}  
\EndFunction
\end{algorithmic}
\end{algorithm}

\subsection{Analysis}

Lemma~\ref{lem: delta-min} below gives a lower bound for the Fr\'{e}chet distance between a curve of size $m$ and another curve of size $n$ such that $m \geq 2n-1$.  An application of this result shows that $\delta_{\min}$ computed in line~\ref{alg:lowerbound} of {\sc Reduce} is a lower bound for $d_F(\sigma^*_{2k-1},\tau[v_1,v_{2k-1}])$ as $|\sigma_{2k-1}^*| = k$.

\begin{lemma}\label{lem: delta-min}
Let $\xi = (p_1,\ldots,p_m)$ and $\zeta = (q_1,\ldots,q_n)$ be any two curves such that $m \geq 2n-1$.  Then, $d_F(\xi,\zeta) \geq \frac{1}{2}\min\bigl\{d(p_i,p_{i-1}p_{i+1}) : i \in [2,m-1]\bigr\}$.
\end{lemma}
\begin{proof}
Let $\mathcal{M}$ be a Fr\'{e}chet matching between $\xi$ and $\zeta$.  Since $m \geq 2n-1$, there must exist an index $c \in [2,m-1]$ such that $\mathcal{M}$ matches $\xi[p_{c-1},p_{c+1}]$ to an edge of $\zeta$.  Let $xy$ be the line segment to which $\xi[p_{c-1},p_{c+1}]$ is matched by $\mathcal{M}$.  Let $r = d_{\mathcal{M}}(\xi[p_{c-1},p_{c+1}],xy)$.  We claim that $r\ge\frac{1}{2}d(p_c, p_{c-1}p_{c+1})$.   Since $d(p_{c-1},x) \leq r$ and $d(p_{c+1},y) \leq r$, by linear interpolation, every point on $xy$ is at a distance no more than $r$ from $p_{c-1}p_{c+1}$.   Let $z$ be the point in $xy$ that is matched with $p_c$ by $\mathcal{M}$.  We have $d(p_c,p_{c-1}p_{c+1}) \leq d(p_c,z) + d(z,p_{c-1}p_{c+1})$.  We have shown that $d(z,p_{c-1}p_{c+1}) \leq r$, so  $d(p_c,p_{c-1}p_{c+1}) \leq 2r$, completing the proof of our claim.  Hence,
$d_F(\xi,\zeta) \geq r \geq \frac{1}{2}d(p_c, p_{c-1}p_{c+1}) \geq \frac{1}{2}\min\bigl\{d(p_i,p_{i-1}p_{i+1}) : i \in [2,m-1]\bigr\}$.
\end{proof}

Lemma~\ref{lem:shortcut} below shows that a curve $(p_1,\ldots,p_{2k-1})$ of size $2k-1$ can be simplified to a curve of size at most $2k-2$ using an error tolerance of $\frac{1}{2}\min\bigl\{d(p_j,p_{j-1}p_{j+1}) : j \in [2,2k-2], \, \text{$j$ is even}\bigr\}$.  By Lemma~\ref{lem:shortcut}, {\sc Compress} always returns a curve of size at most $2k-2$.

\begin{lemma}\label{lem:shortcut}
Let  $\xi=(p_1, \ldots,p_{2k-1})$.  Assume that no three consecutive vertices are collinear.  Let $\hat{\delta}$ be the minimum value such that calling {\sc Simplify}$(\eps, \hat{\delta})$ on $\xi$ produces a curve of size at most $2k-2$.  Then, 
\[
\hat{\delta}\leq \frac{1}{2}\min\bigl\{d(p_j,p_{j-1}p_{j+1}) : j \in [2,2k-2],\, \text{$j$ is even}\bigr\}\le (1+\eps)\hat{\delta}.
\]
\end{lemma}
\begin{proof}
Consider the run of {\sc Simplify}$(\eps, \hat{\delta})$ on $\xi$.  Let $p_{i_1}, p_{i_2}\ldots$ be the vertices in order along $\xi$ at which {\sc Simplify} starts a new segment. 
Note that $i_1=1$ and $i_{j+1} \geq i_j+2$ for all $j$.
    
{\sc Simplify} replaces each subcurve $\xi[p_{i_j},p_{i_{j+1}-1}]$ by a segment.  The simplified curve is a concatenation of these segments.  Since $|\xi| = 2k-1$, there must be an index $j$ such that $i_{j+1} > i_j +2$ so that {\sc Simplify}$(\eps, \hat{\delta})$ simplifies $\xi$ to a curve of size at most $2k-2$.  Let $b = \min\{j : i_{j+1} > i_j + 2 \}$.  Because $i_1=1$ and $i_{j+1} = i_j+2$ for all $j \in [b-1]$, we know that $i_j$ is odd for every $j \in [b]$ and $i_b + 1$ is even. It follows from the choice of $b$ that {\sc Simplify} does not start a new segment at $p_{i_b+2}$.   By the working of {\sc Simplify}$(\eps,\hat{\delta})$, there is a grid point $p$ in $G_{p_{i_b}}$ (defined with $\delta = \hat{\delta}$) such that $S_{i_b+2}[p] \not= \emptyset$.  By the reasoning in the proof of Lemma~\ref{lem:quality}, for any point $x \in S_{i_b+2}[p]$, $d_F(px,\xi[p_{i_b},p_{i_b+2}]) \leq (1+\eps)\hat{\delta}$. 
 On the other hand, by Lemma~\ref{lem: delta-min}, $d_F(px,\xi[p_{i_b},p_{i_b+2}]) \geq \frac{1}{2}d(p_{i_b+1}, p_{i_b}p_{i_b+2})$.  Hence, $(1+\eps)\hat{\delta} \geq 
\frac{1}{2}d(p_{i_b+1}, p_{i_b}p_{i_b+2}) \geq  
\frac{1}{2}\min\bigl\{d(p_j,p_{j-1}p_{j+1}) : j \in [2,2k-2],\, \text{$j$ is even}\bigr\}$.

It remains to argue that $\hat{\delta} \leq \frac{1}{2}\min\bigl\{d(p_j,p_{j-1}p_{j+1}) : j \in [2,2k-2],\, \text{$j$ is even}\bigr\}$.  Let $\delta = \frac{1}{2}\min\bigl\{d(p_j,p_{j-1}p_{j+1}) : j \in [2,2k-2], j\text{ is even}\bigr\}$.  It suffices to show that calling {\sc Simplify}$(\eps,\delta)$ on $\xi$ returns a curve of size at most $2k-2$.   Let $p_{a_1}, p_{a_2},\ldots$ be the vertices in order along $\xi$ at which {\sc Simplify}$(\eps,\delta)$ starts a new segment.  Let $c$ be the even number in $[2,2k-2]$ such that $d(p_c,p_{c-1}p_{c+1}) = \min\bigl\{d(p_j,p_{j-1}p_{j+1}) : j \in [2,2k-2],\, \text{$j$ is even}\bigr\}$.  

Suppose that there is an index $j$ that satisfies $a_j \leq c-1$ and $a_{j+1} > a_j+2$.  Since $\xi[p_{a_j},p_{a_{j+1}-1}]$ is simplified to a line segment, the prefix $\xi[p_1,p_{a_{j+1}-1}]$ is simplified to a curve of size at most $a_{j+1}-2$, i.e., at least one vertex less.  We are done because the size of the whole simplified curve must be at most $2k-2$.  

The remaining possibility is that $a_{j+1} = a_j + 2$ for all $a_j \in [c-1]$. Since $a_1 = 1$ and $c$ is even, we have $a_1 = 1, a_2 = 3, a_3 = 5, \ldots, a_j = c-1$.  It follows that {\sc Simplify} starts a new segment at $p_{c-1}$. 
Let $r=d(p_c, p_{c-1}p_{c+1})$.  

We claim that there is a line segment $xy$ such that $d_F(xy,\xi[p_{c-1}, p_{c+1}])\leq r/2$.  Let $z$ be the nearest point in $p_{c-1}p_{c+1}$ to $p_c$.   So $zp_c$ has length $r$.  Let $z'$ be the midpoint of $zp_c$.  Translate $p_{c-1}p_{c+1}$ by the vector $z'-z$ to obtain a segment $xy$.  Observe that $z' \in xy$ and $d(x,p_{c-1}) = d(y,p_{c+1}) = d_F(z',p_c) = r/2$.  By linear interpolations between $xz'$ and $p_{c-1}p_c$ and between $z'y$ and $p_cp_{c+1}$, we have $d_F(xy,\xi[p_{c-1},p_{c+1}]) \leq r/2$, establishing our claim.

By the choice of $c$, $\frac{1}{2}r = \frac{1}{2}d(p_c,p_{c-1}c_{c+1}) = \delta$.  Then, our claim implies that $xy$ stabs balls of radii $\delta$ centered at $p_{c-1}, p_c, p_{c+1}$ in order.  There is a grid point $p'$ in $G_{p_{c-1}}$ (defined using $\delta$) within a distance $\eps\delta/2$ from $x$.  By a linear interpolation between $p'y$ and $xy$, we know that $p'y$ stabs $\conv(G_{p_{c-1}}), \conv(G_{p_c}), \conv(G_{p_{c+1}})$ in order.  Therefore, $S_{c+1}[p'] \not= \emptyset$ by Lemma~\ref{lem:stab}(i), implying that {\sc Simplify} does not start a new segment at $p_{c+1}$.
Hence, {\sc Simplify} replaces a subcurve $\xi[p_{c-1},p_e]$ for some $e \geq c+1$ by a segment and produces   
a curve of size at most $2k-2$.  This proves that $\delta \geq \hat{\delta}$.
\end{proof}

Recall the curves $\tau_i$ for $i \geq 2k-1$ defined in the comments in lines~\ref{alg:tau-1}, \ref{alg:tau-2}, and \ref{alg:3}.  The next result follows immediately from the working of {\sc Reduce}.

\begin{lemma}
\label{lem:epoch}
For $i \geq 2k-1$, $\sigma_i$ computed by {\sc Reduce} can also be produced by calling {\sc Simplify}$(\eps,\delta_i)$ on $\tau_i$.
\end{lemma}
\cancel{
\begin{proof}
Let $a$ be the largest index in $[2k-1,i]$ such that {\sc Reduce} calls {\sc Compress} when processing $v_{a}$.  Then, the computations performed by this call of {\sc Compress} as well as the subsequent executions of lines~\ref{alg:2-start}--\ref{alg:2-end} of {\sc Reduce} for processing $v_{a+1},\ldots,v_i$
determine line segments in $\sigma_i$ in the same manner as in {\sc Simplify}.  Note that $\delta_a = \delta_{a+1} = \cdots = \delta_i$.
\end{proof}
}

We can also show that $\tau_i$ is a faithful approximation of $\tau[v_1,v_i]$.

\begin{lemma}\label{lem: deal_with_beta}
Assume that $\eps \in (0,1/3]$.  For $i \geq 2k-1$, $d_F(\tau_i,\tau[v_1,v_i]) \leq 2\eps\delta_i$.
\end{lemma}
\begin{proof}
We prove the lemma by induction on $i$.  As $\tau_{2k-1} = \tau[v_1,v_{2k-1}]$, $d_F(\tau_{2k-1},\tau[v_1,v_{2k-1}])$ is zero.  Assume that the lemma is true for some $i-1 \geq 2k-1$.  There are two cases.
		
If $\delta_i = \delta_{i-1}$, then $\tau_i = \tau_{i-1} \circ (v_i)$.  Since $d_F(\tau_{i-1},\tau[v_1,v_{i-1}]) \leq 2\eps\delta_{i-1}$, the last vertex of $\tau_{i-1}$ is within a distance $2\eps\delta_{i-1}$ from $v_{i-1}$.   A linear interpolation between $v_{i-1}v_i$ and the last edge of $\tau_i$ shows that $d_F(\tau_i,\tau[v_1,v_i]) \leq d_F(\tau_{i-1},\tau[v_1,v_{i-1}]) \leq 2\eps\delta_{i-1} = 2\eps\delta_i$.

If $\delta_i > \delta_{i-1}$, then $\tau_i = \sigma_{i-1} \circ (v_i)$.  {\sc Compress} ensures that $\delta_i \geq \eps^{-1}\delta_{i-1}$. 
By Lemma~\ref{lem:epoch} and Theorem~\ref{thm: delta_main}, $d_F(\tau_{i-1},\sigma_{i-1}) \leq (1+\varepsilon)\delta_{i-1}$.  Therefore, $d_F(\sigma_{i-1},\tau[v_1,v_{i-1}]) \leq d_F(\tau_{i-1},\sigma_{i-1}) + d_F(\tau_{i-1},\tau[v_1,v_{i-1}]) \leq (1+3\eps)\delta_{i-1} \leq \eps(1+3\eps)\delta_i \leq 2\eps \delta_i$.  A linear interpolation between $v_{i-1}v_i$ and the last edge of $\tau_i$ 
gives $d_F(\tau_i,\tau[v_1,v_{i}]) \leq d_F(\sigma_{i-1},\tau[v_1,v_{i-1}]) \leq 2\eps\delta_i$.
\end{proof}

Recall that $\sigma^*_i$ is the curve of size $k$ at minimum Fr\'{e}chet distance from $\tau[v_1,v_i]$ and that $\delta_i^* = d_F(\sigma_i^*,\tau[v_1,v_i])$. Also, recall that for $i \geq 2k-1$, $r_i \geq 1$ and $s_i \geq 0$  are integers that satisfy the following inequalities:
\begin{equation}
\frac{1}{\eps^{s_i}}\delta_{\min} 
\leq 
\frac{(1+\varepsilon)^{r_i-1}}{\varepsilon^{s_i}}\delta_{\min}
\leq \delta_i^*
\leq 
\frac{(1+\varepsilon)^{r_i}}{\varepsilon^{s_i}}\delta_{\min}
\leq 
\frac{1}{\eps^{s_i+1}}\delta_{\min}.
\label{eq:def}
\end{equation}
By Lemma~\ref{lem: delta-min}, $\delta_{\min} \leq \delta_{2k-1}^*$.  Also, $\delta_{a}^* \leq \delta_b^*$ for all $a < b$ because a Fr\'{e}chet matching between $\sigma^*_b$ and $\tau[v_1,v_b]$ also matches $\tau[v_1,v_a]$ to a curve of size at most $k$.  Therefore, $\delta_{\min} \leq \delta^*_i$ for $i \geq 2k-1$, which implies that $r_i$ and $s_i$ are well defined for $i \geq 2k-1$.  Note that $r_i$ is an integer between 1 and $\lfloor \log_{1+\eps} (1/\eps)\rfloor$.


	 

\begin{lemma}\label{lem: deal_with_alpha}
For all $i \geq 2k-1$ and $\eps \in \bigl(0,\frac{1}{17})$, {\sc Reduce}$(r_i,\varepsilon,k)$ computes $\delta_i\le (1+8\varepsilon)\delta_i^*$.
\end{lemma}
\begin{proof}
We prove that {\sc Reduce}$(r_i, \eps, k)$ computes $\delta_i \leq \eps^{-s_i}(1+\eps)^{r_i+1}(1-4\eps)^{-1}\delta_{\min}$.
By~\eqref{eq:def}, this is at most $(1+\eps)^2(1-4\eps)^{-1}\delta^*_i \leq (1+8\eps)\delta_i^*$ for $\eps \leq 1/17$.

Consider the case of $i = 2k-1$.  We examine the call {\sc Reduce}$(r_{2k-1},\eps,k)$.  
By Theorem~\ref{thm: delta_main}, calling {\sc Simplify}$(\eps,\delta^*_{2k-1})$ on $\tau_{2k-1} = \tau[v_1,v_{2k-1}]$ produces a curve of size at most $2k-2$.  Then, Lemma~\ref{lem:shortcut} and~\eqref{eq:def} imply that 
$\frac{1}{2}\min\bigl\{d(v_j,v_{j-1}v_{j+1}) : j \in [2,2k-2], \, \text{$j$ is even}\bigr\} \leq 
(1+\eps)\delta^*_{2k-1} \leq \eps^{-s_{2k-1}}(1+\eps)^{r_{2k-1}+1}\delta_{\min}$. 
Since the input parameter $r$ of {\sc Reduce} is equal to $r_{2k-1}$, line~\ref{alg:1-start} of {\sc Reduce} sets $\delta_{2k-2} = \eps(1+\eps)^{r_{2k-1}+1}(1-4\eps)^{-1}\delta_{\min}$.  
As a result, line~\ref{alg:shortcut} of {\sc Compress} covers $\eps^{-i}\delta_{2k-2}$ for $i \geq 1$, which are $\eps^{-j}(1+\eps)^{r_{2k-1}+1}(1-4\eps)^{-1}\delta_{\min}$ for $j \geq 0$.  Therefore, line~\ref{alg:shortcut} of {\sc Compress} must set $\delta_{2k-1}$ to a value no more than $\eps^{-s_{2k-1}}(1+\eps)^{r_{2k-1}+1}(1-4\eps)^{-1}\delta_{\min}$ because this quantity is greater than $\frac{1}{2}\min\bigl\{d(v_j,v_{j-1}v_{j+1}) : j \in [2,2k-2], \, \text{$j$ is even}\bigr\}$.  The base case is thus taken care of.

Consider an index $i > 2k-1$.  We examine the call {\sc Reduce}$(r_i,\eps,k)$.  Define:
\begin{equation}
\Delta = \eps^{-s_i}(1+\eps)^{r_i+1} \delta_{\min}.  
\label{eq:def0}
\end{equation}
Our goal is to prove that $\delta_i \leq (1-4\eps)^{-1}\Delta$.
We rewrite the middle inequalities in \eqref{eq:def} as
\begin{equation}
(1+\eps)^{-2}\Delta \leq \delta_i^* \leq (1+\eps)^{-1}\Delta.
\label{eq:def2}
\end{equation}
Define:
\begin{equation}
a = \max\bigl\{j \in [2k-2,i] : \delta_j \leq \eps(1-4\eps)^{-1}\Delta\bigr\}.
\label{eq:def3}
\end{equation}
The index $a$ is well defined because one can verify that $\delta_{2k-2} = \eps(1+\eps)^{r_i+1}(1-4\eps)^{-1}\delta_{\min} \leq \eps(1-4\eps)^{-1}\Delta$.

By \eqref{eq:def3}, \eqref{eq:def0}, line~\ref{alg:1-start} of {\sc Reduce}, and line~\ref{alg:shortcut} of {\sc Compress}, $\delta_a = \eps^{h}(1-4\eps)^{-1}\Delta$ for some integer $h \geq 1$.  If $a=i$, then $\delta_i = \delta_a \leq \eps(1-4\eps)^{-1}\Delta$, so we are done.  Suppose that $a < i$.  Note that $\delta_{a+1} \not= \delta_a$ by \eqref{eq:def3}, which implies that $\delta_{a+1} > \delta_a$.  So $\tau_{a+1} = \sigma_{a} \circ (v_{a+1})$, and {\sc Compress}$(\tau_{a+1},\eps,k,\delta_{a})$ is called to produce $\sigma_{a+1}$ and $\delta_{a+1}$.  We have
\begin{eqnarray}
d_F(\sigma^*_{a+1},\tau_{a+1}) & \leq &
d_F(\sigma^*_{a+1},\tau[v_1,v_{a+1}]) + d_F(\tau_{a+1},\tau[v_1,v_{a+1}]) \nonumber \\
& = & \delta_{a+1}^* + d_F(\tau_{a+1},\tau[v_1,v_{a+1}]) \nonumber \\
&\leq & \delta_i^* + d_F(\tau_{a+1},\tau[v_1,v_{a+1}]). \label{eq:0}
\end{eqnarray}

If $a+1 = 2k-1$, then $\tau_{a+1} = \tau_{2k-1} = \tau[v_1,v_{2k-1}]$, so $d_F(\tau_{a+1},\tau[v_1,v_{a+1}]) = 0$.  Suppose that $a+1 > 2k-1$.  By Lemma~\ref{lem:epoch} and Theorem~\ref{thm: delta_main}, $d_F(\sigma_{a},\tau_{a}) \leq (1+\eps)\delta_{a}$.  By Lemma~\ref{lem: deal_with_beta}, $d_F(\tau_{a},\tau[v_1,v_{a}]) \leq 2\eps\delta_{a}$.  So $d_F(\tau_{a+1},\tau[v_1,v_{a+1}]) \leq d_F(\sigma_{a},\tau[v_1,v_{a}]) \leq d_F(\sigma_{a},\tau_{a}) + d_F(\tau_{a},\tau[v_1,v_{a}]) \leq (1+3\eps)\delta_{a}$.  Since $\delta_{a} \leq \eps(1-4\eps)^{-1}\Delta$, 
we conclude that
\begin{equation}
d_F(\tau_{a+1},\tau[v_1,v_{a+1}]) \leq \eps(1+3\eps)(1-4\eps)^{-1}\Delta.  
\label{eq:-1}
\end{equation}
Plugging \eqref{eq:-1} into \eqref{eq:0} and using \eqref{eq:def2}, we obtain the following inequality for $\eps < 1/17$:
\begin{equation}
d_F(\sigma_{a+1}^*,\tau_{a+1}) \leq \left(\frac{1}{1+\eps} + \frac{\eps + 3\eps^2}{1-4\eps}\right)\Delta \leq \frac{\Delta}{(1+\eps)(1-4\eps)}.
\label{eq:1}
\end{equation}

Consider the call of {\sc Simplify} in {\sc Compress}$(\tau_{a+1},\eps,k,\delta_{a})$.  By~\eqref{eq:1} and Theorem~\ref{thm: delta_main}, calling {\sc Simplify} on $\tau_{a+1}$ with an error tolerance of $\frac{1}{(1+\eps)(1-4\eps)}\Delta$ will produce a curve of size at most $2k-2$.  Then, by Lemma~\ref{lem:shortcut}, $\frac{1}{2}\min\bigl\{d(v_j,v_{j-1}v_{j+1}) : j \in [2,2k-2], \, \text{$j$ is even}\bigr\} \leq (1+\eps) \cdot \frac{1}{(1+\eps)(1-4\eps)}\Delta = (1-4\eps)^{-1}\Delta$.  Line~\ref{alg:shortcut} of {\sc Compress} ensures that $\delta_{a+1} = \eps^{-t}\delta_{a} = \eps^{h-t}(1-4\eps)^{-1}\Delta$ for the smallest integer $t \geq 1$ such that $\delta_{a+1} \geq \frac{1}{2}\min\bigl\{d(v_j,v_{j-1}v_{j+1}) : j \in [2,2k-2], \, \text{$j$ is even}\bigr\}$.  Hence, $t$ cannot be greater than $h$. But $t$ cannot be less than $h$ because $\delta_{a+1} > \eps(1-4\eps)^{-1}\Delta$ by \eqref{eq:def3}.  So $\delta_{a+1} = (1-4\eps)^{-1}\Delta$.
For every $b \in [a+2,i]$,
\begin{eqnarray*}
d_F(\sigma^*_b,\tau_{a+1} \circ \tau[v_{a+2},v_b]) & \leq &
d_F(\sigma^*_b,\tau[v_1,v_b]) + d_F(\tau_{a+1} \circ \tau[v_{a+2},v_b],\tau[v_1,v_b])  \\
& \leq & \delta_b^* + d_F(\tau_{a+1},\tau[v_1,v_{a+1}])  \\
& \leq & \delta_i^* + \eps(1+3\eps)(1-4\eps)^{-1}\Delta 
\quad\quad\quad~~~(\because \eqref{eq:-1}) \\
& < & (1-4\eps)^{-1}\Delta  \quad\quad\quad\quad\quad~~~~ (\because \eqref{eq:def2} \,\, \text{and} \,\, \eps \leq 1/17) \\
& = & \delta_{a+1}.
\end{eqnarray*}
Then, when processing $v_{a+2}$, Lemma~\ref{lem:epoch} and Theorem~\ref{thm: delta_main} imply that {\sc Reduce} executes lines~\ref{alg:2-start}--\ref{alg:2-end} with $\delta_{a+2} = \delta_{a+1} = (1-4\eps)^{-1}\Delta$ to produce $\sigma_{a+2}$ of size at most $2k-2$.  
Applying this reasoning inductively gives
$\delta_i = (1-4\eps)^{-1}\Delta$.
\end{proof}

Our main result on streaming $k$-simplification are as follows.
 	
\begin{theorem}\label{thm: k_main}
Let $\tau$ be a polygonal curve in $\mathbb{R}^d$ that arrives in a data stream.  Let $\alpha  = 2(d-1)\lfloor d/2 \rfloor^2 + d$.
There is a streaming algorithm that, for any integer $k \geq 2$ and any $\eps \in (0,\frac{1}{17})$, maintains a curve $\sigma$ such that $|\sigma| \leq 2k-2$ and $d_F(\tau, \sigma)\le (1+\varepsilon) \cdot \min\{d_F(\tau,\sigma') : |\sigma'| \leq k\}$.  It uses $O\bigl((k\eps^{-1}+\varepsilon^{-(\alpha+1)})\log \frac{1}{\varepsilon}\bigr)$ working storage and processes each vertex of $\tau$ in $O(k\eps^{-(\alpha+1)}\log^2\frac{1}{\eps})$ time for $d \in \{2,3\}$ and $O(k\eps^{-(\alpha+1)}\log\frac{1}{\eps})$ time for $d \geq 4$.
\end{theorem}
\begin{proof}
Let $\delta^* = \min\{d_F(\tau,\sigma') : |\sigma'| \leq k \}$.
By 
Lemma~\ref{lem: deal_with_alpha}, if we launch $\lfloor \log_{1+\eps/10} (10/\eps)\rfloor = O(\frac{1}{\eps}\log\frac{1}{\eps})$ runs of {\sc Reduce} using $\eps/10$ instead of $\eps$, there is a run that outputs a curve $\sigma$ such that $d_F(\tau,\sigma) \leq (1+\eps/10)(1+8\eps/10)\delta^* \leq (1+\eps)\delta^*$.  

Each run of {\sc Reduce} uses $O(k)$ space to maintain the output curve and uses another $O(\eps^{-\alpha})$ working storage by Theorem~\ref{thm: delta_main}.  This gives a total of $O\bigl((k\eps^{-1}+\varepsilon^{-(\alpha+1)})\log \frac{1}{\varepsilon}\bigr)$ working storage over all runs.  The vertex processing time is the highest when {\sc Compress} is called to simplify a curve of size $2k-1$.  By Theorem~\ref{thm: delta_main}, this call takes $O(k\eps^{-\alpha}\log\frac{1}{\eps})$ time for $d \in \{2,3\}$ and $O(k\eps^{-\alpha})$ time for $d \geq 4$.  Summing over all $O(\frac{1}{\eps}\log\frac{1}{\eps})$ runs gives the total processing time.
\end{proof}

\section{Conclusion}

We present streaming algorithms for the $\delta$-simplification and $k$-simplification problems.  Both use little working storage and process the next vertex in the stream efficiently.  Moreover, they offer provable guarantees on the size of the output curve and the Fr\'{e}chet distance between the input and output curves.  There are some natural research questions.  In this work, the vertices of the output curve may not be a subset of those of the input curve.  What guarantees can be obtained if this requirement is imposed?  Can other curve proximity problems be solved efficiently in the streaming setting?

\bibliographystyle{plain}
\bibliography{ref}

\end{document}